\documentclass[12pt,oneside,english]{amsart}
\usepackage[T1]{fontenc}
\usepackage{lmodern}
\usepackage[letterpaper]{geometry}

\usepackage{amsthm}
\usepackage{amssymb}
\usepackage{setspace}
\usepackage{esint}
\usepackage[round,authoryear]{natbib}
\usepackage{hyperref}
\hypersetup{
	colorlinks=false,
	citecolor=green
}
\makeatletter

\theoremstyle{definition}
\newtheorem{defn}{\protect\definitionname}
\newtheorem{ax}{\protect\axiomname}
\newtheorem*{ax*}{\protect\axiomname}
\newtheorem{ass}{Assumption}
\theoremstyle{plain}
\newtheorem{thm}{\protect\theoremname}
\theoremstyle{plain}
\newtheorem{cor}{\protect\corollaryname}
\theoremstyle{plain}
\newtheorem{prop}{Proposition}
\theoremstyle{definition}
\newtheorem*{example*}{\protect\examplename}
\theoremstyle{plain}
\newtheorem{lem}{\protect\lemmaname}
\usepackage{lmodern}
\makeatother

\usepackage{babel}
\providecommand{\axiomname}{Axiom}
\providecommand{\definitionname}{Definition}
\providecommand{\examplename}{Example}
\providecommand{\lemmaname}{Lemma}
\providecommand{\corollaryname}{Corollary}
\providecommand{\theoremname}{Theorem}

\usepackage{hyperref}
\hypersetup{
	colorlinks=false,
	citecolor=green
}

\begin{document}
	\title{Correlation Concern}
	\author[Andrew Ellis]{Andrew Ellis}
	\date{May 2021}
	\thanks{The author wishes to thank Michele Piccione, Larry Epstein, and seminar/conference audiences at Bristol, Copenhagen, the Lisbon Meetings, Royal Holloway, Southern Denmark, Toulouse, UCL, and Warwick and  for helpful comments and conversations.
	}
	\begin{abstract}
		In many choice problems, the interaction between several distinct variables determines the payoff of each alternative. I propose and axiomatize a model of a decision maker who  recognizes that she may not accurately perceive the correlation between these variables, and who takes this into account when making her decision. She chooses as if she calculates each alternative's expected outcome under multiple possible correlation structures, and then evaluates it according to the worst expected outcome. 
	\end{abstract}
	\maketitle

	\section{Introduction}
	In many decision problems, the overall payoff of each alternative depends on multiple, distinct variables; for instance, the return of a stock portfolio depends on the return of the underlying stocks. Understanding the correlations between these underlying objects is difficult both conceptually and econometrically.%
\footnote{Throughout, I use the term ``correlation'' interchangeably with the more accurate ``joint distribution.''} Concern about her own (or another agent's) lack of a good understanding of these interdependencies can materially change the behavior of a decision maker (DM). 
	
	An individual may choose an index fund over the corresponding stocks because she does not recognize their connection and is uncertain about the correlation between the stocks. A financial institution may choose a suboptimal loan portfolio in order to pass a stress test that ensures it is not subject to too much systematic risk. A principal may offer a simple contract to ensure that it is robust to the agent's perception of the correlations between the payoffs it offers, the agent's own information, and the private information and actions of other agents. 
	
	I propose and axiomatize a model of a DM who recognizes that she may not accurately perceive the correlation between the underlying sources of uncertainty, and who takes this into account when making her decision. 
	The DM expresses preferences over portfolios of assets whose payoffs all depend on a common state space,  known to the modeler but not necessarily the DM.
	As shown in \cite{EllisPiccione2017} (henceforth, EP), she misperceives correlation when she has a strict preference over two portfolios that always  yield  the same consequence.
	I propose axioms corresponding to a preference for alternatives with payoffs that do not depend on the correlation. The main result shows that a DM's behavior satisfies them if and only if she can be represented as if she considers a set of possible correlation structures between the actions and evaluates each alternative by the worst expected utility in that set.
	
	To illustrate how correlation concern and misperception are identified, consider a DM offered the choice between holding either an S\&P 500 index-tracking fund or a portfolio of the 500 underlying stocks of the S\&P 500 (in the right proportions and without transaction costs).
	When she expresses a strict preference for one or the other, she reveals that she is uncertain about the correlation between the stocks.
	An agent who  recognizes and dislikes her potential misperception opts for the simpler choice of the tracker fund.
	
	The novel \emph{Negative Uncorrelated Independence} axiom captures this concern about correlation in the alternatives. 
	Intuitively, it says that if she prefers an alternative where correlation matters over one where it does not, then introducing potential correlation to both without making either better or worse does not lead to a preference reversal.
	Formally, I capture this by introducing lotteries and weakening the independence axiom. If the DM prefers a stock portfolio $P$ to an individual stock $s$, then she also prefers a lottery between $P$ and a second portfolio $P'$ to a lottery between $s$ and $P'$.
	I also follow EP in imposing a Weak Monotonicity Axiom that requires that whenever one portfolio is better than another for every possible joint distribution over outcomes, that portfolio is preferred. 
	
	These two axioms, jointly with other standard axioms, are necessary and sufficient for the DM's preference to have a \emph{Correlation Concern Representation (CCR)}.
	The representation consists of a set of   joint distributions over the payoffs of assets, each consistent with an underlying distribution when restricted to a given asset. She evaluates each portfolio according to its worst expected utility according to one of these distributions.
	When the DM understands sufficiently rich subsets of assets, then her perception of possible correlation structures can be uniquely recovered from her choices.

	Such a DM is not probabilistically sophisticated, even on the larger state space considered by EP.
	This allows it to capture 
	a number of   behaviors exhibited by a sophisticated agent who recognizes  that she does not understand the interdependencies. For instance, she may refuse to take either side of a trade that she does not understand.
	As noted by \cite{EpsteinHalevy2018}, with expected utility it is impossible that ``awareness of the complexity of her environment
	and self-awareness of her cognitive limitations lead the decision-maker
	to doubt that her wrong beliefs are correct.''
	
	I apply the model to study asset pricing when agents have a CCR.
	The agent understands the connections between the assets she is endowed with, such as those from her domestic stock market, but may misunderstand its connection with other stocks. 
	I show that this can lead to unexploited arbitrage opportunities, even in equilibrium, as well as failure to purchase stocks not in the endowment at a range of prices. The simple application links misunderstanding to several anomalies for rational expectation, subjective expected utility models. In particular, it formally suggests that misunderstanding can play a role in explaining home bias \citep{FrenchPoterba1991}, limited participation \citep{MankiwZeldes1991}, and failure of no-arbitrage conditions \citep{Fleckensteinetal2014}.

    The paper proceeds as follows.
    The next subsection reviews particularly relevant related literature.
    Section \ref{sec:framework} introduces the formal framework.
    Section \ref{sec:axioms} begins with an illustration that probabilistic beliefs rule out reasonable behavior when correlation is misunderstood.
    It then states and discusses the behavioral axioms and other assumptions necessary for identification.
    Section \ref{sec: representation} presents the formal representations and discusses its uniqueness properties and interpretation.
    Section \ref{sec:discussion} discusses how particularly relevant special cases fit into the model and presents the asset pricing application.

	\subsection*{Related Literature}
	There is ample experimental evidence for correlation misperception, including \cite{enkezimmerman2013}, \cite{EysterWeizsacker2010}, 
	\cite{SalantRubinstein2015}, and \cite{HossainOkui2018}.
	Experimental evidence that agents are aware of and concerned about correlation can be found in \cite{EpsteinHalevy2018}. They study an environment with explicit uncertainty about the correlation between two events whose joint realization determines the payoff of a bet. A majority of subjects are inconsistent with a probabilistic model of correlation, and a majority of that majority at least weakly prefer bets that do not depend on correlation.
	Indirect evidence comes from theoretical study of asset markets.
	In this context, \cite{jiang2016correlation}, \cite{liu2017correlation}, and \cite{huang2017limited} consider such a model and show it explains some stylized facts including under-diversification and limited participation in the market.
	
	\cite{EpsteinSeo2010, EpsteinSeo2015} explore axiomatically the consequences of introducing ambiguity in the classic exchangeable model of de Finetti. An exogenous product state space describes the outcome of a sequence of experiments that are indistinguishable and possibly related but not identical. They provide a model where the DM perceives ambiguity about the relationship between experiments. The functional form, especially in the 2010 paper, is similar to the one I consider, but acts depend exogenously on a collection of experiments.

	\cite{EpsteinHalevy2018} consider a related model. They argue that, in the setting above, an ambiguity averse agent may prefer bets on only a single experiment to bets that depend on multiple experiments. As noted above, they conduct an experiment that confirms that a number of subjects have this preference.
	
		\cite{heo2020}	provides a different perspective on  DM averse to acts that depend on different components of an (objectively given) product state space, or issues. He argues that such a DM may strictly prefer an act that depends on a single issue to a mixture of that act with an equally good act that depends on a different issue. As a consequence, she may violate the classic uncertainty aversion axiom. 
		Adapted to this setting, NUI requires that a DM who prefers a multi-issue act to a single issue act does not reverse that preference when both are mixed with a common third act. Moreover, a DM with a correlation concern representation satisfies the uncertainty aversion axiom, though we do not explicitly impose it. These approaches provide complementary perspectives on the issue.
	
To apply the above papers that take a product state as given, the modeler must know what the DM perceives the state space to be and observe her ranking of acts on this state space. In settings where she misperceives her options, this may be more difficult. For instance in the thought experiment  in Section \ref{sec: illustration}, their model would imply that the DM can be observed to rank bets on impossible events like ``temperature is greater than $0^\circ C$ but lower than $32^\circ F$.''		
	
	Concern for for robustness has other applications in mechanism design. For instance, \cite{Carroll2017} considers a seller uncertain about the correlation between the buyer's values of different goods that can be bundled. The CCR captures the behavior of such a principal nicely. 	Another strand of the behavioral mechanism design literature  focuses instead on a principal concerned that the agent does not correctly understand the game that she is playing. Most notably, \cite{Li2017}  considers obviously strategy proof mechanisms: mechanisms played correctly even agents who do not understand the relationship between the other agents' actions and information with her own payoff.
	This is conceptually connected to the results herein, but in the CCR an agent evaluates each mechanism with the ``worst'' beliefs about the relationship to others. One can model this by considering incomplete preferences and maintaining independence instead of maintaining completeness and NUI.

	\cite{LevyRazin2015amb} and \cite{Laohakunakornetal2018} consider an agent who is exposed to information from multiple sources. She considers all priors ``close enough'' to a benchmark when making her decision. As in this paper, she consider the worst of these priors when evaluating acts. The benchmark is full independence .
	

	\section{Primitives}
	
	\label{sec:framework} There is a set $\mathcal{A}$ of \emph{actions}, with
	typical elements $a,a_{i},b,b_{i}$. Each action results in an \emph{outcome} or consequence in
	$X=\mathbb{R}$, with typical elements $x,y,z$. This outcome is determined by a
	\emph{state of the world} drawn from the finite set $\Omega$, with typical elements $\omega,\omega'$. I interpret the state space $\Omega$ as a description of the
	``objectively possible'' joint realizations
	of the outcomes of any set of actions, against which the DM's subjective
	perceptions of joint realizations are evaluated.
	
	A map $\rho:\mathcal{A}\times\Omega\rightarrow X$ describes the relationship
	between actions, states, and outcomes, with the action $a$ yielding the
	outcome $\rho(a,\omega)$ in state $\omega$.  The set $\mathcal{A}$ includes every Savage act, i.e. for any $f:\Omega \rightarrow X$ there is an action yielding the outcome $f(\omega)$ in state $\omega$ for every 
	$\omega \in \Omega$. Slightly abusing notation, 
	write $a(\omega)$ for $\rho(a,\omega)$ and $x$ for an action that yields $x\in X$ in every state. 
	
	The set $\mathcal{F}$ of \emph{action
		profiles} (or profiles) is such that each element is a finite vector of	actions for which the order does not matter -- i.e., a multiset of actions.
	A profile that consists of taking
	the $n$ actions $a_{1},...,a_{n}$ is denoted $\langle a_{1},...,a_{n}\rangle$
	or $\langle a_{i}\rangle_{i=1}^{n}$. To save notation, the
	range of indexes is omitted when the number of actions is unimportant, i.e.  $\langle
	a_{i}\rangle$ instead of $\langle a_{i}\rangle_{i=1}^{n}$. An agent who
	chooses the profile $\langle a_{i}\rangle_{i=1}^{n}$ receives the outcomes of
	all $n$ actions $a_{1},...,a_{n}$, that is, she receives $\sum_{i=1}^{n}%
	a_{i}\left(  \omega\right)  $ is state $\omega$.
	
	The DM chooses by maximizing a preference relation $\succsim$ over probability distributions on
	$\mathcal{F}$ having finite support, the set of which is denoted by
	$\Delta\mathcal{F}$. A typical
	element of $\Delta\mathcal{F}$ is $p=\left(  p_{1},\langle a_{i}^{1}%
	\rangle;...;p_{m},\langle a_{i}^{m}\rangle\right)  $, interpreted as the
	lottery where profile $\langle a_{i}^{j}\rangle$ occurs with probability
	$p_{j}$.
	As usual, the symbol $\sim$ denotes indifference and $\succ$
	strict preference. 
	The set of lotteries over actions, $\Delta \mathcal{A}= \{p\in \Delta \mathcal{F}:p(\langle a_i \rangle_{i=1}^n)>0 \text{ only if }n=1\}$, plays an important role in the axioms, and includes as a subset  lotteries over outcomes.
	 It will be convenient to write $p(\langle a_{i}\rangle)>0$ for the set of profiles $\langle
	a_{i}\rangle$ in the support of $p$. Given $p,q\in\Delta\mathcal{F}$, a mixture $\alpha
	p+(1-\alpha)q$, $\alpha\in\lbrack0,1]$, is the lottery in $\Delta\mathcal{F}$
	in which the probability of each profile in the support of $p$ and $q$ is
	determined by compounding the probabilities in the obvious way.

	Endow $\Delta\mathcal{F}$ with the weak* topology for the space $\mathcal{F}$ endowed with metric $d$ defined as follows. Let $X^*=\{ \langle x \rangle: x \in X \}$. The metric $d$ is discrete on $\mathcal{F} \setminus X^*$ and agrees with the usual metric on $X$ on $X^*$. Formally, for any $\langle a_{i}\rangle_{i=1}^{n},\langle b_{i}\rangle_{i=1}^{n'} \in \mathcal{F}$, 
	$d(\langle a_{i}\rangle_{i=1}^{n},\langle b_{i}\rangle_{i=1}^{n'})=1$ when  $ \left\{ \langle a_{i}\rangle_{i=1}^{n} ,\langle b_{i}\rangle_{i=1}^{n'} \right\} \not \subset X^*$ and $\langle a_{i}\rangle_{i=1}^{n}\neq \langle b_{i}\rangle_{i=1}^{n'}$, $d(\langle a_{i}\rangle_{i=1}^{n},\langle b_{i}\rangle_{i=1}^{n'})=0$ if $\langle a_{i}\rangle_{i=1}^{n}=\langle b_{i}\rangle_{i=1}^{n'}$, and  $d(\langle x \rangle, \langle y \rangle)=  |x-y|$. According to $d$, a sequence of profiles converges only if it is eventually constant, or every profile therein is a single, constant outcome, and the sequence of outcomes converges.%
	\footnote{A sequence of profiles $(F_n)$ $d$-converges to the profile $ G$ only if there exists $N$ so that either $F_n=G$ for all $n>N$ or $F_n=\langle x_n \rangle $ for all $n>N$, $G=\langle y \rangle$,  and the sequence $(x_n)_{n>N}$ approaches $y$ in the usual sense. 
	}

	\section{Foundations}
	
	\label{sec:axioms}
	This section begins by presenting a thought experiment illustrating the novel behavior the model captures. It then introduces and discusses the axioms on the DM's preference relation invoked by the main result. Finally, the assumptions necessary for the identification result are introduced.
	\subsection{Illustration of behavior }
	\label{sec: illustration}
	Consider a DM  choosing between bets that depend on $\tau$,
	tomorrow's high temperature. The DM can have either $\$100$ or the sum of the
	outcomes of bets $b_{C}$ and $b_{F}$, where $b_{C}$ pays $\$100$ if $\tau$ is
	less than $30$ degrees Celsius ($\$0$ otherwise) and $b_{F}$ pays $\$100$ if
	$\tau$ is at least $86$ degrees Fahrenheit ($\$0$ otherwise). 
	On another occasion with the same weather forecast, the DM must choose between $b$, which pays pays $\$100$ if $\tau$ is
	less than $30$ degrees Celsius and $-\$100$ otherwise, and the combination of $b_C$ and $-b_F$, which is a ``short'' position on $b_F$ that  pays $-\$100$ if
	$\tau$ is at least $86$ degrees Fahrenheit ($\$0$ otherwise).
	Formally, the DM makes a choice from each of the sets  $\{\langle 100 \rangle,\langle b_C, b_F \rangle\}$ and $\{\langle b \rangle,\langle b_C, -b_F \rangle\}$.
	As $30^{\circ}$
	Celsius equals $86^{\circ}$ Fahrenheit, a DM who knows this and easily
	converts Fahrenheit to Celsius expresses indifference in both choices. 
	However, a DM who does not know exactly how to convert from one unit to the other may not exhibit such indifference and
	reasonably express $\langle 100 \rangle \succ \langle b_C, b_F \rangle$ and $\langle b \rangle \succ \langle b_C, -b_F \rangle$.

	A probabilistic approach to correlation can capture only one of the two preferences. For instance, a risk-averse DM may express $\langle 100 \rangle \succ \langle b_C, b_F \rangle$ because she believes that $b_C$ and $b_F$ may be positive correlated and thus riskier than $\$100$ for sure. However, that same DM believes that $b_C$ and $-b_F$ are negatively correlated and thus less risky than $b$.
	Hence under expected utility, $\langle 100 \rangle \succ \langle b_C, b_F \rangle$ implies $\langle b_C, -b_F \rangle \succ \langle b \rangle$.
	The next section outlines behavior consistent with the thought experiment and equivalent to a representation  of a DM concerned about correlation.
	
	\subsection{Basic Axioms}
	
	The first set of assumptions are  either invoked by EP or closely related. The first two are standard.
	
	\begin{ax}
		[Weak order, WO]\label{ax: WO}\label{ax: first} The preference relation $\succsim$ is complete and transitive.
	\end{ax}
	\begin{ax}
		[Continuity, C]\label{ax: Cont} 
		For any sequences $p_n,q_n \in \Delta(\mathcal{F})$,\\ if $p_n \rightarrow p$, $q_n \rightarrow q$, and $p_n \succsim q_n$ for all $n$, then $p \succsim q$.
	\end{ax}
	The DM can compare any pair of lotteries over action profiles, and her pairwise comparisons are sufficiently consistent with each other to form an ordering. 
	Moreover, the ranking is sufficiently continuous. 
	According to the topology introduced in Section \ref{sec:framework}, it is continuous in probability (and only probability) when the lotteries involve profiles that have non-constant payoffs, but it is continuous in the usual sense when restricting to lotteries over constant payoffs. 
	In particular $(1,\langle x_n \rangle ) \rightarrow (1,\langle x \rangle )$ whenever $x_n \rightarrow x$.
	
	As in EP, misperception of correlation occurs whenever the DM violates Monotonicity. In this context, Monotonicity holds if $$\left(  p\big(\langle a_{i}\rangle_{i=1}^{n}\big),\langle\sum_{i=1}%
	^{n}a_{i}(\omega)\rangle\right)  _{p\left(  \langle a_{i}\rangle\right)  >0} \succsim \left(  q\big(\langle a_{i}\rangle_{i=1}^{n}\big),\langle\sum_{i=1}%
	^{n} a_{i}(\omega)\rangle\right)  _{q\left(  \langle a_{i}\rangle\right)  >0}$$ for every $\omega \in \Omega$ implies that $p \succsim q$. 
	EP suggest the following weakening. It allows the behavior in the thought experiment, yet rules out other violations that cannot be attributed to misperception of correlation, such as expressing $\langle 49,50 \rangle \succ \langle 100 \rangle$.
	
	Stating the axiom formally requires some notation. For any finite subset of actions $\{c_{1},...,c_{n}\}=C\subset
	\mathcal{A}$, the set of all \emph{plausible realizations} of $C$ equals
	\[
	range(c_{1})\times range(c_{2})\times...\times range(c_{n}).
	\]
	A vector of outcomes $\vec{x}$ is a \emph{plausible realization of the lotteries
		$p$ and $q$} if  it is a plausible realization of the set of all the actions
	included in profiles that are assigned positive probability by either $p$ or
	$q$.
	For a plausible realization $\vec{x}$ of $p$ and $q$, \emph{$p$ induces
		the lottery 
	\[p_{\vec{x}}=
	\left(  p\big(\langle a_{i}\rangle_{i=1}^{n}\big),\langle\sum_{i=1}%
	^{n}x^{a_{i}}\rangle\right)  _{p\left(  \langle a_{i}\rangle\right)  >0}%
	\]}
	in which the constant action yielding the outcome $\sum_{i=1}^{n}x^{a_{i}}$
	has probability $p\left(  \langle a_{i}\rangle_{i=1}^{n}\right)  $. 
	
	\begin{ax}
		[Weak Monotonicity, WM]\label{ax: PD} For any $p,q\in\Delta\mathcal{F}$,\\
		if	$p_{\vec{x}} \succsim  q_{\vec{x}}$ (respectively, $p_{\vec{x}} \succ  q_{\vec{x}}$) for every plausible realization $\vec{x}$
		of $p$ and $q$, then $p\succsim q$ (respectively, $p\succ q$). 
	\end{ax}

	In words, if the DM prefers the lottery induced by $p$ better than to that
	induced by $q$ for any of their plausible realizations, then she prefers $p$
	to $q$.
	To illustrate, consider three actions, $a,b,c$, and consider a DM who must choose between the profile containing $a$ and $b$, denoted $\langle a,b \rangle$, and the one containing only $c$, denoted $\langle c \rangle$. In this case, if the minimum payoff of action $a$ added to the minimum payoff of $b$ exceeds the maximum payoff of $c$, then $\langle a,b \rangle$ is preferred to $\langle c \rangle$. This implies that $\langle 49,50 \rangle \succ \langle 100 \rangle$, but it \emph{does not require} that  $\$100$ for sure is preferred to the combination of $b_C$ and $b_F$ in the thought experiment, since the latter could also yield $\$200$ (or $\$0$).
	For additional interpretation and discussion, see EP.
	
	Finally, the DM makes comparisons among actions without difficulty.
	\begin{ax}[Simple Monotonicity, SM]
		For any $a,b \in \mathcal{A}$ , \\
		if $a(\omega)  \geq b(\omega)$
		for all $\omega \in \Omega$, then $\langle a \rangle  \succsim \langle b\rangle $, strictly whenever the inequality is strict for each state.
	\end{ax}
	
	This is a standard monotonicity condition, but only applied to profiles consisting of a single action.
	These profiles correspond to (standard) Savage acts.
	Restricted to these objects the DM behaves as a standard utility maximizer.
	
	\subsection{Preference for Avoiding Correlation}
	I now introduce the key axiom that reflects a DM who dislikes exposure to correlation. It weakens the Independence Axiom, which holds that for any $p,q,r \in \Delta \mathcal{F}$ and $\alpha \in (0,1]$, $p \succsim q$ implies that $ \alpha p +(1-\alpha)r \succsim \alpha q +(1-\alpha)r$.
	
	\begin{ax}[Negative Uncorrelated Independence, NUI]
	\label{ax: NUI}
	\label{ax: last}
		For any $p,q,r\in\Delta\mathcal{F}$ and $\alpha \in (0,1]$:\\
		if $p \succsim q$ and  $q \in \Delta \mathcal{A}$, then $\alpha p +(1-\alpha)r \succsim \alpha q +(1-\alpha)r$.
	\end{ax}
	
	Relative to the usual axiom, NUI requires that mixing does not reverse preference only when the worse of the lotteries attaches zero probability to non-singleton profiles.
	To illustrate, suppose that $\langle a,b \rangle$ is indifferent to $\langle c \rangle$. While the DM's evaluation of $\langle a,b \rangle$ depends on her perception of the correlation between the two stocks, $\langle c \rangle$ does not involve any correlation at all. Hence, the indifference reflects that the absence of exposure to correlation by $\langle c \rangle$ exactly offsets a better expected outcome from $\langle a,b \rangle$.
	After mixing both with an arbitrary lottery $r$ over profiles, evaluating both alternatives requires computing correlation. Moreover, neither alternative becomes objectively better: the mixture changes the expected utility of both profiles in the same way for a given correlation between actions  by standard usual independence axiom arguments. Nonetheless, the ``simplicity'' advantage that $\langle c \rangle$ had is lost. 
	A DM who dislikes thinking about correlations should then prefer the mixture of $\langle a,b \rangle$ and $r$ to the mixture of $\langle c \rangle$ and $r$.
	
	NUI, and the other axioms, implies that violations of independence do not occur when comparing lotteries over actions. Moreover, any violation of independence favors a lottery over profiles. Fixing a correlation structure, the mixing with a common third profile does not make either alternative ``differentially better'' using the standard logic behind the Independence Axiom.
	However, it can hedge against the possibility of different correlation structure being realized. As with the standard justification of the independence axiom, there is no expected utility based reason to reverse preference after mixing. 
	
	A similar logic to NUI underlies axioms that appear in \cite{GMMS2010}, \cite{Dillenberger2010}, and \cite{Cerreia-Vioglioetal2015}, with certain alternatives playing the role of uncorrelated profiles. The first assumes that the DM defaults to certainty: if an act is not objectively better than a lottery, then subjectively it should not be better. The other two assume that the DM  prefers sure outcomes: if a lottery is not objectively better than a sure outcome, then a mixture of the lottery is better than a mixture of the sure outcome.

	\subsection{Understanding and Richness}
	The previous axioms are necessary and sufficient for the representation theorem.
	However, additional assumptions allow for unique identification of a ``coarsest'' state space and a unique set of beliefs.
	This subsection presents these definitions, which also appear in EP.
	I defer to that paper for a detailed discussion.
	
	The definition of understanding extends the logic of Weak Monotonicity. A DM
	who perceives the correlations within a subset of the actions correctly rules
	out some plausible realizations. Specifically, if she understands the
	correlations of the actions in a set $C$, then she should consider irrelevant
	any plausible realization of $p$ and $q$ that fails to ``synchronize'' the outcomes for the actions in $C$ as for the
	joint occurrences that are determined by $\Omega$. That is, she should only
	consider a plausible realization $\vec{x}$ if there exists $\omega\in\Omega$
	such that $x^{a}=a(\omega)$ for all $a\in C$;  call any such plausible
	realization \emph{$C$-synchronous}.
	
	\begin{defn}
		\label{def:understand}The preference $\succsim$ \emph{understands}
		$C\subseteq\mathcal{A}$ if for any $p,q\in\Delta\mathcal{F}$, $p\succsim q$
		whenever $p_{\vec{x}}\succsim q_{\vec{x}}$ for all $C$-synchronous plausible
		realizations $\vec{x}$ of $p$ and $q$.
	\end{defn}
	
	In order to identify the main parameters of the representations, I 
	assume that each action belongs to a suitably diverse, understood set of
	actions. A sufficient condition is that this understood set is rich, defined
	as follows.
	
	\begin{defn}
		A set $B\subset\mathcal{A}$ is \emph{rich} if for any function $f:\Omega\rightarrow X$, there exists $a\in	B$ with $a(\omega)=f(\omega)$ for all $\omega \in \Omega$.
	\end{defn}
	
	This is a slightly stronger definition than that used by EP which held that a set is rich if it contains every function measurable with respect to a given algebra. 
	I can now state the assumption.
	
	\begin{ass}[Non-Singularity] Each $a \in\mathcal{A}$ belongs to a rich, understood subset
		of actions.
	\end{ass}
	
	Non-Singularity is in the traditional \cite{Savage1954} assumption that the
	domain of preference contains all possible acts. It is a joint assumption on
	both the preference $\succsim$ and the set $\mathcal{A}$.
	It ensures that understood sets of
	
	\section{Representation}
	\label{sec: representation}
	\subsection{Correlation Concern Representation}
	A DM who misperceives correlation perceives additional uncertainty beyond that captured by $\Omega$.
	To capture it, I follow EP by  replacing the objective state space $\Omega$ with a subjective state space rich enough to capture the DM's perception of uncertainty. To do so, let
	$\Omega^{\mathcal{A}}=\prod_{a\in\mathcal{A}}\Omega$ be the Cartesian product
	where one copy of $\Omega$ is assigned to each action in $\mathcal{A}$,
	$\Sigma^{\mathcal{A}}=\otimes_{a\in\mathcal{A}}\sigma\left(  a\right)  $ be
	the product $\sigma$-algebra on $\Omega^{\mathcal{A}}$ where $\sigma(a)$ is the smallest algebra by which $a$ is measurable, $\Omega^{a}$ be the
	copy of $\Omega$ assigned to $a\in\mathcal{A}$, and for any $\vec{\omega}%
	\in\Omega^{\mathcal{A}}$, $\omega^{a}$ be the component of $\vec{\omega}$ in
	$\Omega^{a}$.
	The DM is represented as if she considers events in the larger state space $\Omega^{\mathcal{A}}$.
	Every $\vec{\omega}\in\Omega^{\mathcal{A}}$ determines a joint realization of
	the outcomes of the corresponding actions, with $a$ yield $a(\omega^a)$, $b$ yielding $b(\omega^b)$, and so on. Hence, all additional uncertainty
	corresponds to the perception of correlations. 
	\begin{defn}
	\label{def: CCR}
		A preference $\succsim$ has a \emph{Correlation Concern Representation (CCR)} if
		there exists
		\vspace{-.3cm}
		\begin{itemize}
			\item a continuous, strictly increasing 
			utility index \(u:X \rightarrow \mathbb{R}\),
			\item 	a  probability measure $\mu$ on $\Omega$, 
			\item and a set of finitely additive probability measures \(\Pi\) on
			\(
			\left(\Omega^{\mathcal{A}},\Sigma^{\mathcal{A}} \right) \) whose marginals agree with $\mu$: for any \(a \in \mathcal{A}\) and all \(\pi \in \Pi\),
			\[\int_{\Omega^{\mathcal{A}}} u(a(\omega^a)) d\pi =\int_{\Omega} u(a(\omega)) d\mu\]
		\end{itemize} 
		\vspace{-.3cm}
		\noindent
		such that for any \(p,q \in \Delta \mathcal{F}\),
		\(p\succsim q \iff V(p) \geq V(q)\) where
		\[V \left( p \right) = \min_{\pi \in \Pi}    \int_{\Omega^\mathcal{A}} \mathbb{E}_{p(\langle a_i \rangle)} \left[ u \left(\sum_{i=1}^n   a_i(\omega^{a_i}) \right) \right] d\pi. \]
	\end{defn}

	While she acts as if she maximizes  expected utility with probability measure  $\mu$ when comparing individual actions, she does not when comparing profiles. When evaluating them, she considers a set of possible joint distributions possible, represented by the set $\Pi$. A lottery over profiles is evaluated by its expectation according to the measure that minimizes the resulting expected utility, as in \cite{GilboaSchmeidler1989}. Consequently, she acts as if she is averse to uncertainty about correlations but not about the returns of individual actions.
	Without a subjective state space, the model has a number of precedents.
	Most notably, \cite{EpsteinSeo2010,EpsteinSeo2015} use non-expected utility models with an objective product state space as above to capture ambiguity about the relationship between distributions of different variables.
	In contrast,  the mapping between the objects of choice and the product state space is subjective and derived from the representation.

	The axioms introduced above are necessary and sufficient for the DM's behavior to be represented by a CCR.

		\begin{thm}
		\label{thm:basic_rep}
		The preference $\succsim$ satisfies  Weak Order, Continuity, Weak Monotonicity, Simple Monotonicity, and Negative Uncorrelated Independence if and only if it has a correlation concern representation.
	\end{thm}
	
		The properties of the representation relate naturally to the axioms imposed. Action Monotonicity and NUI imply that the DM acts a standard subjective expected utility maximizer  when dealing with single action profiles.
		Weak Monotonicity allows the DM to misunderstand correlation between actions, as captured when the representation has a $\pi \in \Pi$ so that $\pi(\omega^a\neq \omega^b)>0$ for some $a,b \in \mathcal{A}$. 
		It nonetheless implies that the DM ignores ``implausible'' outcomes, so each probability measures in the set $\Pi$ attaches zero probability to such outcomes, e.g. $\langle b_C,b_F \rangle$ yielding $\$300$ or $-\$400$, as captured by the subjective state space $\Omega^\mathcal{A}$.  The representation captures recognition and dislike that the DM does not understand correlation by allowing $\Pi$ to be a non-singleton set. As shown by Corollary \ref{cor: EP Thm 1}, this follows from NUI.
		
		\begin{cor}[\cite{EllisPiccione2017}]
	\label{cor: EP Thm 1}
	The preference \(\succsim\) satisfies Weak Order, Continuity, Weak Monotonicity, Simple Monotonicity, and Independence if and only if 
it has a CCR where $\Pi$ is a singleton.
	\end{cor}
	
As noted in Section \ref{sec:axioms}, NUI is the key behavioral difference from EP. NUI implies a concern for correlation not captured by the independence axiom. A consequence of the violation of independence is a strict preference for randomization. Formally, there may exist  lotteries so that $p \sim q$ but $\frac12 p +\frac12 q \succ q$. This occurs in much of the ambiguity aversion literature and underlies the logic of the uncertainty aversion axiom. In the CCR model, preference for randomization follows from NUI. Here, the randomization is explicit since the ordering of the horse-race and roulette wheel is reversed.

While CCR provides a simple and interpretable model equivalent to the axioms, the set of priors is not unique, nor is the state space. The state space is canonical, in the sense that any joint distribution over the returns of actions can be expressed by a probability measure on it, but it has a very large number of dimensions. This makes it flexible but potentially hard to apply. The next subsection shows that under non-singularity, a unique state space and set of priors can be identified.

The formal proof of Theorem \ref{thm:basic_rep} can be found in the appendix. The following outlines the main arguments showing sufficiency of the axioms for the representation. NUI implies independence for lotteries over individual actions, which allows identification of a utility index and the marginal probability measure $\mu$ by following \cite{AnscombeAumann1963}.
Each  lottery over profiles is mapped to a real valued act on the product state space $\Omega^\mathcal{A}$.
A utility value is assigned to each by setting the utility equal to that of an action equivalent, with
Weak Monotonicity and Continuity insuring that one exists.
By NUI, this utility function is is homogeneous of degree one and superlinear, but is defined only on a convex subset of acts on $\Omega^\mathcal{A}$.
The key step extends it to all bounded, measurable functions while maintaining the above properties.
Arguments following \cite{GilboaSchmeidler1989} then establish the result.

	\subsection{Identification and Rich Representation}
	
	This section proposes a more tractable special case of CCR with a more parsimonious state space that has a unique set of priors.
	The DM acts as if she undergoes the following procedure. 
	First, she groups together certain actions that she understands as per
	Definition \ref{def:understand} into an
	\emph{understanding class}.
	The classes are revealed from her choices.
	Then, she forms beliefs about the return within each class.
	Finally, she constructs a set of beliefs about the potential relationships across classes.
	When comparing any two profiles, evaluate she each according to the worst of its possible expected utilities from these beliefs.
	The main result of this section shows that this representation exists whenever the DM exhibits the behavior implied by the correlation concern representation as long as there exist sufficiently rich subsets of actions that the DM understands. Moreover, the components of the representation are unique for most utility indexes, unlike the CCR.
	
    The understanding classes are contained in a \emph{correlation cover
	$\mathcal{U}$ for $\succsim$}; formally, $\mathcal{U}$  is a collection of subsets of $\mathcal{A}$ so that $\mathcal{U}$ covers $\mathcal{A}$, each $C \in \mathcal{U}$ is understood by $\succsim$, and no $C\in\mathcal{U}$ contains a distinct $C^{\prime}\in\mathcal{U}$.
	The cover $\mathcal{U}$ is \emph{rich} if each $C \in \mathcal{U}$ is rich.
	A rich correlation cover $\mathcal{U}$ is \emph{coarsest} if any other rich  cover $\mathcal{U}'$ has the property that for any $C' \in \mathcal{U}'$, there exists $C \in \mathcal{U}$ so that $C' \subset C$.
	Examples include $\{\{a\}:a \in \mathcal{A}\}$ and $\{\mathcal{A}\}$; the former is always a correlation cover but never rich, while the latter is the coarsest rich correlation cover provided that $\mathcal{A}$ is understood.
	EP show that there exists a unique coarsest rich correlation cover under non-singularity.
	
	Beliefs are defined on the subjective state space $\Omega^{\mathcal{U}}=\prod_{C\in\mathcal{U}}\Omega$, 
	with the $C$-coordinate denoted
	by $\Omega^{C}$, endowed with the product $\sigma$-algebra 
	$\Sigma^{\mathcal{U}}=\otimes_{C\in\mathcal{U}}\Sigma$.
	All actions in the same class are assigned to	the same copy of $\Omega$, so the possible joint realizations within a class	are identical to the objective ones. 
	Given a state $\vec{\omega}\in\Omega^{\mathcal{U}}$ and a class $C \in \mathcal{U}$, $\omega_{C}$ denotes the component of $\vec{\omega}$ assigned to $C$, and given an event $E \in \Sigma^{\mathcal{U}}$, $E_C$ is the projection of $E$ onto the $C$ component.
	
	\begin{defn}
	\label{def: rich CCR}
		A preference $\succsim$ has a \emph{Rich CCR} $(u,\mu,\mathcal{U},\Pi)$ if
		\vspace{-.5cm}
		
		\begin{itemize}
			\item  \(u:X \rightarrow \mathbb{R}\) is a continuous, strictly increasing
			utility index,
			\item $\mu$ is a  probability measure  on $\Omega$,
			\item $\mathcal{U}$ is the coarsest rich correlation cover for $\succsim$,
			\item and \(\Pi\)  is a
			set of finitely additive probability measures on
			\(
			\left(\Omega^{\mathcal{U}},\Sigma^{\mathcal{U}} \right) \)
			where  the marginal of every $\pi \in \Pi$ on every $C \in \mathcal{U}$ equals $\mu$
		\end{itemize}
		\vspace{-.3cm}
		such that for any $p',q' \in \Delta \mathcal{F}$, $ p' \succsim q' \iff V(p')\geq V(q')$ where
		\[
		V \left( p \right) = \min_{\pi \in \Pi}    \int_{\Omega^\mathcal{U}} \mathbb{E}_{p(\langle a_i \rangle)} \left[ u \left(\sum_{i=1}^n   a_i \left(\omega^{C(a_i)} \right) \right) \right] d\pi  
		\]	for any $C: \mathcal{A} \rightarrow \mathcal{U}$  with $a \in C(a)$ for all $a\in \mathcal{A}$.
	\end{defn}
	
	If one assumes Non-Singularity , then the CCR axioms hold if and only if the DM has a Rich CCR.
	
	\begin{thm}\label{thm:CC_rep}
		Under Non-Singularity, the preference $\succsim$ satisfies Weak Order, Continuity, Weak Monotonicity, Simple Monotonicity, and Negative Uncorrelated Independence if and only if it has a Rich Correlation Concern Representation. Moreover, $\mu$ and $\mathcal{U}$ are unique, $u$ is unique up to a positive affine transformation,  and if $u$ is not a polynomial, then $\Pi$ is unique.	\end{thm}
    
    A DM represented by a  rich CCR behaves according to the same axioms as one who has a CCR. Non-singularity plays an analogous role to the usual assumption in decision theory that every act is conceivable and ranked by the DM. This allows construction of a richer representation in which the parameters are unique. As standard, the DM's risk preference is identified from her preference over lotteries. But here her perception of possible correlations is identified uniquely as well, unless her risk attitude is such that she does not care about it. This allows cleaner interpretation of the set of probability measures in the representation. This will be exploited in the subsequent section to analyze some special cases of interest in the applied theory literature.
    
    \section{Applications and extensions}
	\label{sec:discussion}
	This section begins by exploring two polar special cases of the model.
	In the first, an agent considers the true correlation structure but is nonetheless concerned that there may be a different one.
	In the second, the agent treats all misperceived objects as independent, but is concerned that her heuristic may be incorrect.
	Finally, an asset pricing example is presented that relates a number of anomalies to misperception of the relation between the underlying stocks.

	\subsection{Concern about Complexity}
	Consider a DM who prefers to avoid ``complex'' profiles whose evaluation requires her to compute correlations in favor ``simple'' profiles that would yield the same outcome if all correlations are understood. In its simplest manifestation, this DM expresses \[\langle c \rangle \succsim \langle a,b\rangle\] whenever $a(\omega)+b(\omega)=c(\omega)$ for all states $\omega$.
	If the DM understands correlations, then the two alternatives always yield the same outcome. Although a rational DM expresses indifference between the two, the single action profile has no exposure to correlation	whereas the multi-action profile might. 
	
	I capture this in general with the following axiom.
	\begin{ax*}[Complexity Aversion]
		For any $p \in  \Delta\mathcal{A}$ and $q\in\Delta\mathcal{F}$, if \[
		\left(p(\langle a \rangle),\langle a(\omega) \rangle \right)_{p(\langle a \rangle)>0}  \succsim \left( q(\langle a_i \rangle), \sum_i a_i(\omega) \right)_{q(\langle a_i)>0}
		\]
		for all $\omega \in \Omega$, then $p \succsim q$.
	\end{ax*}
	Complexity Aversion requires that the DM opts for a lottery over actions over any lottery over action profiles whenever the lottery over actions yields at least as good an outcome in every state of the world according to the true structure.
	To interpret the axiom, note that the lottery $p$ involves no proper profiles, only individual actions.
	In this sense, $p$ is less complex than $q$. 
	Moreover,  $p$  yields at least as good an outcome as $q$ according to the true model.
	The axiom says that these two advantages are sufficient for the DM to prefer $p$ to $q$.
	Note that Full Monotonicity implies Complexity Aversion which in turn implies Simple Monotonicity.
	
	The following proposition explores the implications of this axiom.
	
	\begin{prop}
		\label{prop: pref_simple}
		Let $\succsim$ have a rich CCR $\left(u,\mu,\mathcal{U},\Pi\right)$ where $u$ is not a polynomial and $\mathcal{U}$ is finite.
		Then, $\succsim$ satisfies Complexity Aversion if and only if $\mu \in \Pi$, i.e. there exists $\pi' \in \Pi$ so that  $\pi'(E)=\mu \left(\bigcap_{C\in \mathcal{U}} E_C \right)$ for every $E \in \Sigma^\mathcal{U}$.
	\end{prop}
	
	In words, a DM with a rich CCR is complexity Complexity Averse  if and only if the true correlation structure belongs to her set of possible priors. 
	 Complexity Aversion  cannot be exhibited by a DM who takes a probabilistic approach to correlation unless she correctly perceives every lottery. Such a DM necessarily overvalues some misunderstood profiles while overvaluing others. For instance, consider a risk averse DM who misunderstands the connection between two stocks $a$ and $b$. She overvalues the profile $\langle a,b \rangle$ whenever she underestimates their correlation, and undervalues it when she over estimates the correlation. But if she overestimates the correlation between $a$ and a long position on stock $b$, she underestimates the correlation between $a$ and a short position on $b$.
	Consequently, she overvalues one if and only if she undervalues the other.

	\subsection{Concern about  correlation neglect}
	A common way of modeling correlation neglect is for agents to treat any two actions or signals with unknown correlation as independent.
	This section considers a DM who applies this heuristic, yet realizes that it may lead to incorrect inference.
	This has been applied, for instance, by \cite{LevyRazin2015amb}.
	For this subsection only, I restrict attention to only two understanding classes to simplify exposition.
	
	Stating the axiom and result requires a few preliminary definitions.
	A bet is an action $b$ with exactly two possible consequences. 
	A pair of bets $(b_1,b_2)$ is potentially misperceived if there exist rich, understood sets $C_1,C_2$ with $b_1 \in C_1$ and $b_2 \in C_2$ where $C_1 \bigcup C_2$ is not understood.
	A preference relation $\succsim$ is risk averse if whenever $p \in \Delta X$  Second-Order Stochastically Dominates $q \in \Delta X$, $p \succsim q$, strictly whenever the domination is strict.
	
	I can now state the key axiom of this subsection.
	\begin{ax*}[Default to Independence]
		For any collection of pairs of potentially misperceived bets $(b_1^i,b_2^i)_{i=1}^m$,
		if $p_j^i = \left( p_j^i,x_{ij};(1-p_j^i),y_{ij} \right) \sim b_j^i$ and  $range(b_j^i)=\{x_{ij},y_{ij}\}$ for every $j=1,2$ and $i=1,\dots,m$,
		then the lottery
		$$\sum_{i=1}^m\alpha_i
		\left( p^i_1p^i_2, x_{i1}+x_{i2};p^i_1(1-p^i_2), x_{i1}+y_{i2};(1-p^i_1)p^i_2, y_{i1}+x_{i2};(1-p^i_1)(1-p^i_2), y_{i1}+y_{i2} \right)$$ 
		is weakly preferred to $ (\alpha_i,\langle b_1^i,b_2^i \rangle)_{i=1}^m$
		for any $\alpha_1,\dots,\alpha_m\geq 0$ with $\sum_{i=1}^m \alpha_i=1$.
	\end{ax*}
	To interpret the axiom, consider $m=1$ and pick a pair of misperceived bets $(b_1,b_2)$ so that $b_1$ returns either $x$ or $0$ and $b_2$ pays either $y$ or $0$.
	If the lotteries $p=(p,x;(1-p),0)$ and $q=(q,y;(1-q) 0)$ are indifferent to $b_1$ and $b_2$ respectively, then the DM acts as if $b_1$ pays off with probability $p$ and $b_2$ with probability $q$.
	If $p$ and $q$ are independent, then holding both gives the lottery $l^I=(pq,x+y;p(1-q),x;(1-p)q,y;(1-p)(1-q),0)$.
	If the DM thinks $b_1$ and $b_2$ are independent, then she should be indifferent between $l^I$ and $\langle b_1,b_2 \rangle$.
	Instead, the axiom requires that $l^I$ is weakly preferred to $\langle b_1,b_2 \rangle$.
	That is, she values the profile less than she would have if she were certain that $b_1$ and $b_2$ are independent. 
	The axiom extends this case to many pairs of bets.
	
	The axiom captures the behavior in question.
	
	\begin{prop}
	\label{prop: independent benchmark}
		Let $\succsim$ be risk-averse and have a rich CCR $\left(u,\mu,\mathcal{U},\Pi\right)$ where  $\#\mathcal{U}=2$.
		Then, $\succsim$ satisfies Default to Pairwise Independence if and only if  then there is a measure $\mu^2 \in \Pi$ so that $\mu^2(E_{C} \bigcap F_{C'})= \mu(E) \mu(F)$ for every $E ,F\in \Sigma$ and $C,C'\in \mathcal{U}$.
	\end{prop}

    In the representation, the DM considers the possibility that all actions in different understanding classes are independent. Although she may consider other correlation structures, the possibility that the misperceived actions are independent leads her to prefer the sum of two equivalent lotteries that she is certain are independent.
    
    The heuristic of treating objects with unknown correlation as independent has some appeal as a benchmark, especially statistically. The resulting probability measure is the one that maximizes entropy, subject to getting the marginal distribution of each object correct. Using entropy to measure informativeness, this measure makes the ``least'' use of information about which the DM is uncertain.
    
	\subsection{Asset pricing example}
	I turn to some simple implications for asset pricing.
	This subsection focuses on an exchange economy with a representative agent whose preference over asset portfolios has a rich CCR.
	The agent groups the assets she is endowed with into one understanding class. In equilibrium, she does not purchase any assets in a different understanding class for a  range of prices. This non-participation can lead to neglected arbitrage opportunities and a price premium for assets she understands.
	
	The equilibrium parallels many aspects of that in seminal works by \cite{DowWerlang1992} and \cite{EpsteinWang1994}. The key novelty is the relationship between non-participation and indeterminacy with the trader's understanding of the assets.
	In contrast to those papers, indeterminacy and non-participation occur for a generic endowment, provided that the endowed assets belongs to one understanding class.
	Thus it suggests correlation concern as a possible mechanism that explains a number of anomalies for the standard rational asset pricing model, including home bias \citep{FrenchPoterba1991}  and failures of no-arbitrage conditions \citep{Fleckensteinetal2014}.

	Consider an exchange economy with a representative trader, two states $\Omega=\{A,B\}$, and a single consumption good. The trader has available three risky assets, $\alpha$, $\beta$, and $\beta'$. Each is an Arrow security with positive return in states $A$, $B$, and $B$, respectively.
	    She also trades a safe asset, which serves as numeraire, that returns a unit of the consumption good regardless of the state.
   The trader has a strictly concave, strictly increasing utility index  who is endowed with $a $ units of $\alpha$, $b$ units of $\beta$, $0$ of $\beta'$ and $c^*$ units of the safe asset. Assets are traded in a competitive market with prices for securities $\alpha,\beta,\beta'$ given by $p_\alpha, p_\beta,p_{\beta'}$ respectively.
Her utility index $u$ is twice-differentiable with $u'(x)>0>u''(x)$ for every $x>0$.
	
	The action of holding  $x$ units of asset $\gamma$ is denoted $x \gamma$, and the action holding $c$ units of the safe asset is denoted $c$. The trader chooses an asset portfolio in her budget set at the market prices. Her preference over asset portfolios has a Correlation Concern Representation $U(\cdot)$ with two distinct understanding classes, $\mathcal{U}_1$ and $\mathcal{U}_2$, and set of probability measure $\Pi$.  She understands the connection between assets $\alpha$ and $\beta$:  $x \alpha , y \beta \in \mathcal{U}_1 $ for any $x,y$. She does not necessarily understand the connection between either $\alpha $ or $\beta $ and  $\beta'$: $z\beta'\in \mathcal{U}_2$ for any $z$.
	
	With slight abuse of notation, her marginal probability measure $\mu(\cdot)$ assigns $\mu \in (0,1)$ probability to state $A$, and $\pi \in \Pi$ if and only if it has marginals that agree with $\mu(\cdot)$ and $\pi(\omega_1=A,\omega_2=B)\in [r_*, r^*] \subset [0,\min\{ \mu,1-\mu\}]$. 
	The parameters determine how the trader reacts to potentially unknown correlation. If $r_*=r^*=0$, then the agent perceives correlation correctly. Otherwise, she misperceives it. 
	If $r_* <r^*$, then she is concerned about her potential misperception.
	In the language of this section, she satisfies Complexity Aversion if and only if $r_*=0<r^*$, and satisfies Default to Independence if and only if $r_*\leq \mu(1-\mu(B))\leq r^*$.

	First, consider the trader's demand function. Let $V(x,y,z,c)$ be the trader's utility from the asset profile $\langle x \alpha, y \beta, z \beta', c \rangle $. Then, $V(x,y,z,c)$ equals
	\begin{equation}
	    \min_{r\in [r_*,r^*]}  \mu u(x+c)+(1-\mu) u(y+z+c)+rK(x,y,z,c)
	    \label{eq: asset utility}
	\end{equation}
	for $K(x,y,z,c)$ equal to 
	\[
	u(x+z+c)+u(y+c)-u(x+c)-u(y+z+c).
	\]
	As a minimum of strictly concave functions $V$ is strictly concave, so $(x^*,y^*,z^*,c^*)$ is an optimum if and only if $\lambda (p_\alpha,p_\beta,p_{\beta'},1) $ belongs to the subdifferential of $V$, $\partial V(x,y,z,c)$, evaluated at $(x^*,y^*,z^*,c^*)$.
	We are interested in determining  when the trader demands $0$ of asset $\beta'$. Simple calculations show that $v \in \partial V(x,y,0,c)$ if and only if 
	\[
	v=\left(\begin{array}{c}
	 \mu  u'( x +c)\\
	(1-\mu) u'(y+c)\\
	(1-\mu) u'(y+c)+r\left[u'( x +c)-u'( y +c) \right]\\
	\mu  u'( x +c)+(1-\mu) u'(y+c)
	\end{array}
	\right)^T.
	\]
	for some $r \in [r_*, r^*]$ that minimizes (\ref{eq: asset utility}).
	That is, $v$ is a vector whose first, second, and last components equal the derivatives of those variables. The derivative of $V$ in the third component is undefined, leading to the indeterminacy in that component.
	For a portfolio $(x,y,0,c)$ with  $x \neq y$, any $r \in [r_*,r^*]$ minimizes the utility function. Hence, the bundle is demanded at price vector $p^*$ if and only if $$\lambda^{-1} p^*=  \left( \mu u'(x+c),(1-\mu)u'(y+c),(1-\mu) u'(y+c)+r\left[u'( x +c)-u'( y +c) \right]\right)$$
	for some $r \in [r_*,r^*]$ and where $\lambda=\left[\mu  u'( x +c)+(1-\mu) u'(y+c) \right]^{-1}$. Hence, there is typically an interval of prices of $\beta'$ for which the trader demands no units of it.

	Now, turn to the general competitive equilibrium of the single trader exchange economy. In any equilibrium, she consumes her endowment, $(a,b,0,c^*)$. Prices adjust so that she is willing to do so. 
	I calculate these prices. 
	Since the safe asset is numeraire,  
	\begin{equation*}
	   (p_\alpha,p_\beta)= \left(
	 \frac{\mu  u'( a +c^*)}{\mu  u'( a +c^*)+(1-\mu) u'(b+c^*)},
	 \frac{(1-\mu) u'(a+c^*)}{\mu  u'( a +c^*)+(1-\mu) u'(b+c^*)}
	\right)
	\label{eq: price understood assets}
	\end{equation*}
	as with subjective expected utility.
	However $p_{\beta'}$ can be any member of the set 
\[
\left\{p_\beta+r\frac{u'(a+c^*)-u'(b+c^*)}{\mu  u'( a +c^*)+(1-\mu) u'(b+c^*)}: r \in [r_*,r^*] \right\}.
\]
	When $a \neq b$, $p_{\beta} $ need not equal $p_{\beta'}$ and there can be a non-degenerate interval of equilibrium prices for asset $\beta'$ at which the trader neither goes long ($z>0$) nor short ($z<0$) on it. Moreover, this is the case for a large set of endowments of the other two assets. The only other requirement is that  $a+c,b+c >0$, so that the derivative of $u$ is well-defined.

	\newpage
	\appendix
	
	\section{Proofs}
	
	Throughout, I sometimes write $p \alpha q$ instead of $\alpha p+(1-\alpha)q$ for two lotteries $p,q$.

	\begin{lem}
		\label{lem:continuity} Under Axioms \ref{ax: first}-\ref{ax: last}, the set $\{\alpha\in\lbrack0,1]:\alpha	p+(1-\alpha)q\succsim  \alpha p'+(1-\alpha)q'\}$ is closed for all $p,q,p',q'\in\Delta\mathcal{F}$.
	\end{lem}
	\begin{proof}
		Follows from observing that $\alpha_n	p+(1-\alpha_n)q \rightarrow \alpha p +(1-\alpha)q$ according to the weak* topology whenever $\alpha_n \rightarrow \alpha$. See e.g. Theorem 15.3 of \cite{aliprantisborder06}.
	\end{proof}
	\begin{lem}
		\label{lem:action_IA}
		 Under Axioms \ref{ax: first}-\ref{ax: last}, for any $p,q,r\in\Delta\mathcal{F}$ and $\alpha \in (0,1]$, 
		if $r \in \Delta \mathcal{A}$, then $p \succsim q \iff \alpha p +(1-\alpha)r \succsim \alpha q +(1-\alpha)r$.
	\end{lem}
	\begin{proof}
		Fix any $r \in \Delta \mathcal{A}$ and $\alpha \in (0,1]$.
		
		First, $p \succ q \implies \alpha p +(1-\alpha)r \succ \alpha q +(1-\alpha)r$.   If not, then $p\succ q$ and $\alpha q +(1-\alpha)r \succsim \alpha p +(1-\alpha)r$ for some $p,q$.
		Axioms WO, C, and WM imply there exists $p' \in  \Delta \mathcal{A}$ with $p' \sim p$.
		By NUI and WO, $q \alpha r \succsim p \alpha r \succsim p' \alpha r$.
		Let \[\tau = \sup \{ \beta \in [0,1]: q \beta r \succsim p' \beta r \}.\]
		By Lemma \ref{lem:continuity}, $q\tau r \succsim p' \tau r$.
		Since $p' \tau r \in  \Delta \mathcal{A}$, $(q\tau r)\beta q \succsim (p'\tau r)\beta q$ and $(q\tau r)\beta p' \succsim (p'\tau r)\beta p'$ by NUI for any $\beta \in (0,1)$.%
		\footnote{The remainder of the argument is based on one that appears in \cite{shapleybaucells98}.}
		
		Observe $(p'\tau r)\frac{1}{1+\tau} q =(q\tau r)\frac{1}{1+\tau} p'$.
		By WO,  \[
		q\frac{2\tau }{1+\tau} r=(q\tau r)\frac{1}{1+\tau} q \succsim 
		(p'\tau r)\frac{1}{1+\tau} q
		=(q\tau r)\frac{1}{1+\tau} p'  \succsim 
		(p'\tau r)\frac{1}{1+\tau} p'=p'\frac{2\tau }{1+\tau} r.\]
		Thus $\tau \geq 2\tau /(1+\tau)$, which can hold only if $\tau =1 $ or $\tau =0$.
		Since $\tau \geq \alpha >0$,  $\tau=1$. But then $q=q\tau r \succsim p' \tau r=p'$ by Lemma \ref{lem:continuity}, implying that $q\succsim p$ since $p \sim p'$. 
		
		Now, $p \sim q \implies \alpha p +(1-\alpha)r \succsim \alpha q +(1-\alpha)r$.\footnote{A symmetric argument obtains indifference.}
		Fix any $p \sim q$, and pick a lottery $x \in \Delta X$ s.t. $x \succsim p_{\vec{y}}$ for all plausible realizations $\vec{y}$ of $p$.
		For all $\epsilon\in (0,1]$,  $p \epsilon x \succ q$ by WM and WO, and so by the above, $(p \alpha r) \epsilon (x \alpha r)= (p \epsilon x)\alpha r \succsim q \alpha r$.
		By C, $p \alpha r \succsim q \alpha r$ as well. 
		
		Now, one has that $p \succ q \implies p\alpha r \succ q \alpha r$ and $p \succsim q \implies p \alpha r \succsim q \alpha r$. The second, combined with the contrapositive of the first and completeness,  is that  $p \succsim q \iff p \alpha r \succsim q \alpha r$.
		This completes the proof.
	\end{proof}
	
	\begin{lem}
		\label{lem:converse_action_IA}
		 Under Axioms \ref{ax: first}-\ref{ax: last}, if $p,q \in \Delta A$ and $r \in \Delta F$, then \[p \succsim  q \iff \alpha p+(1-\alpha)r\succsim \alpha q+(1-\alpha)r. \]
	\end{lem}
	\begin{proof}
		Pick $r' \in \Delta A$ with $r' \sim r$; this exists by WO, C,and WM.
		Observe $ p\alpha r \sim p\alpha r'$ and $ q \alpha  r' \sim q\alpha  r$ by Lemma \ref{lem:action_IA}.
		Then,   $p \succsim q$ $\iff$ $ p \alpha  r' \succsim q \alpha  r'$ also by Lemma \ref{lem:action_IA}. 
		By WO, $p \alpha r' \succsim q \alpha r'$ $\iff$ $ p \alpha  r \succsim q \alpha  r$.
	\end{proof}
	
	\begin{lem}
		\label{lem:utility_index}
		 Under Axioms \ref{ax: first}-\ref{ax: last}, an affinely unique, continuous utility index $u$ for lotteries over $X$ exists.
	\end{lem}
	
	\noindent
	This follows from the above lemmas and \cite{Grandmont1972}.
	
	\begin{proof}[Proof of Theorem \ref{thm:basic_rep}]
		Necessity is trivial. For sufficiency, assume that $\succsim$ satisfies the axioms.
		By the above lemmas, 
		Normalize so that $range(u) \supset [-1,1]$, and let $0$ denote a lottery with $u(0)=0$.
		Recall that $B_0 (X,\Sigma)$ is the simple, real-valued, $\Sigma$-measurable functions on $X$.
		For any $p \in \mathcal{F}$ define $f_p \in B_0(\Omega^\mathcal{A},\otimes_{a \in \mathcal{A}} \sigma(a))$ so that \[f_p = \vec{\omega} \mapsto \sum_{p( \langle a_{i}\rangle)>0)}   p\big(\langle a_{i}\rangle_{i=1}^{n}\big) u \left(\sum_{i=1}^{n}a_i(\omega^{a_{i}}) \right) . \] 
		By WM, $f_p \geq f_q$ (resp., $f_p \gg f_q$)  implies that $p \succsim q$ (resp., $p\succ q$).
		
		Define $W=\{f_{p}:p\in\Delta(\mathcal{F})\}$, noting that $W$ is convex. For
		$\phi$ in $W$, define $\tilde{I}(\phi)=\int u(x)dr$ for some $p\in\Delta(\mathcal{F})$
		s.t. $f_{p}=\phi$ and a lottery $r$ over $X$ satisfying $r\sim p$. 
		Such an $r$ exists for every $p$ by Weak Monotonicity, Completeness, and Continuity, so
		$\tilde{I}$ is well-defined.
		Denote $f_0=0$. 
		
		The function $\tilde{I}$ has the following properties for any $\alpha \in (0,1]$, $\phi =f_q$, $\psi = f_r$ and $g=f_p$ with $p \in \Delta A $, $q,r \in \Delta F$, $x \in \Delta X$ with $u(x)=x$ and $f_x=x$.		
		\begin{itemize}
			\item $\tilde{I}(\cdot)$ is normalized: $\tilde{I}(x)=x$ by construction. Based on this, we  abuse notation slightly by also identifying $\tilde{I}(\theta)$ with a lottery over $X$ that yields utility $\tilde{I}(\theta)$ for any $\theta \in W$.
			\item $\tilde{I}(\cdot)$ is monotone: $\phi \geq \psi$ implies $\tilde{I}(\phi)\geq \tilde{I}(\psi)$. Follows from WM.
			
			\item $\tilde{I}(\cdot)$ is action invariant: 
			$\tilde{I}( \alpha \psi +(1-\alpha)g ) = \alpha \tilde{I}(\psi)+(1-\alpha)\tilde{I}(g)$.
			To see this, note that by Weak Monotonicity and Continuity, there exists  $z\in \Delta X$ such that $r \sim z$; for this $z$, $\int u(x)dz= \tilde{I}(\psi)$. 
			By  Lemma \ref{lem:converse_action_IA}, $r \alpha p \sim z \alpha p$. 
			Since $z\alpha p \in \Delta X$ and $f_{r \alpha p} =\alpha \psi +(1-\alpha)g$, $\tilde{I}(\alpha \psi +(1-\alpha)g) = \int u(x) d[z\alpha p]= \alpha \tilde{I}(\psi) + \alpha  \tilde{I}(g)$.
			
			\item $\tilde{I}(\cdot)$ is concave: $\tilde{I}(\alpha  \phi +(1-\alpha) \psi)\geq \alpha \tilde{I}( \phi) +(1-\alpha) \tilde{I}(\psi)$.
			To see this, note  $q \sim  \tilde{I}(\psi)$. By NUI, $r \alpha q \succeq r \alpha \tilde{I}(\psi)$.
			Since $\tilde{I}$ is action invariant and normalized, \[\tilde{I}( \alpha \phi+(1- \alpha) \tilde{I}(\psi))=\alpha \tilde{I}(\phi)+(1-\alpha) \tilde{I} (\psi).\]
			
			\item $\tilde{I}(\cdot)$ is Homogeneous of Degree 1: $\tilde{I}(\alpha \psi)=\alpha \tilde{I}(\psi)$. 
			This follows from action invariant and normalized.	
			
			\item $\tilde{I}(\cdot)$ is supnorm continuous. Suppose $\phi^n \rightarrow \phi$ for some sequence $\phi^n$ and $\phi$ that belong to $W$. Let $x^n = \max_{\vec{\omega}} [\phi^n(\vec{\omega})-\phi(\vec{\omega}) ]$ and $y^n = \max_{\vec{\omega}} [\phi^n(\vec{\omega})-\phi(\vec{\omega}) ]$.
			Pick $\kappa,\kappa' \in \Delta X$ with $u(\kappa)=1$ and $u(\kappa')=-1$.
			For $n$ large enough that $|x_n|,|y_n|<1$, \[
			\tilde{I} \left(y_n [\frac12\phi +\frac12 f_{\kappa'}] +(1-y_n) [\frac12\phi +\frac12 0]\right) 
			\leq \tilde{I} \left(\frac12\phi^n +\frac12 0 \right)
			\]	and \[
			\tilde{I} \left(\frac12\phi^n +\frac12 0 \right) 
			\leq \tilde{I} \left(x_n[\frac12 \phi + \frac12 f_{\kappa'}]+ (1-x_n) [\frac12\phi+ \frac12 0] \right)
			\] since $\tilde{I}$ is monotone. 
			By continuity and that $$z_n \left[ \frac12 q +\frac12  \kappa''  \right] +(1-z_n)\left[ \frac12 q +\frac12 0 \right] \rightarrow \frac12 q +\frac12 0$$ for any $\kappa'' \in \Delta X$ in the weak* topology whenever $z_n \rightarrow 0$,  $\tilde{I}(\phi \frac12 0) = \lim \tilde{I}(\phi^n \frac12 0)$. Action independence then gives the result.
		\end{itemize}
		
		Given the above and that $0 \in W$, extend $\tilde{I}$ to the cone generated by $W$ (which is simply denoted by $\tilde{I}$ and $W$ for convenience) using the identity that $\tilde{I}(\alpha \phi)=\alpha \tilde{I}(\phi)$. Clearly, all the above properties are maintained.
		The set $W$ is a convex cone contained in the vector space $B(\Omega^\mathcal{A},\otimes_{a \in \mathcal{A}} \sigma(a)) = W^*$, the bounded, $\otimes_{a \in \mathcal{A}} \sigma(a)$-measurable functions.
		Extend $\tilde{I}$ to $W^*$ as follows.
		
		For any $x \in W^*$, define \[
		I(x)= \sup \left\{ \tilde{I}(w):  x \geq w,\ w\in W \right\}.
		\]
		The function $I$ inherits the following properties from $\tilde{I}$:
		\begin{itemize}
			\item $\phi \in W$ implies $I(\phi)=\tilde{I}(\phi)$: 
			First,  $\phi \in W$ and $\phi \leq \phi$ immediately imply that  $I(\phi) \geq \tilde{I}(\phi)$. 
			Second, $w \leq \phi$ immediately yields $ \tilde{I}(w) \leq \tilde{I}(\phi)$ by monotonicity of $\tilde{I}$, so $\tilde{I}(\phi) \geq I(\phi)$ also.
			
			\item $I$ is concave: fix $\phi,\psi \in W^*$ and $\lambda \in (0,1)$.
			For any $\epsilon>0$, there exist $w_1,w_2 \in W$ with $w_1 \leq \phi$ and $w_2 \leq \psi$ such that $\tilde{I}(w_1)>I(\phi)-\epsilon/2$ and $\tilde{I}(w_2)>I(\psi)-\epsilon/2$. Now, $\lambda w_1 +(1-\lambda) w_2 \leq \lambda \phi +(1-\lambda)\psi $. Then, $I(\lambda \phi +(1-\lambda) \psi) \geq\tilde{I}(\lambda w_1 +(1-\lambda) w_2)\geq \lambda \tilde{I}(w_1)+(1-\lambda) \tilde{I} (w_2)>\lambda I(\phi)+(1-\lambda) I(\psi)-\epsilon$. Letting $\epsilon$ go to zero establishes the result.
			
			\item 
			$I$ is Monotone and $I(x) < \infty $ for all $x$:
			Monotone is trivial. Since $y=\min_{\omega} x(\omega)\leq x$ belongs to $W$, $I(x)> \tilde{I}(y)=y$. Letting $z= \max_{\omega} x(\omega)$, for any $w \in W$ with $x \geq w$,
			$z \geq  w$. Thus $z = \tilde{I}(z) \geq \tilde{I}(w)$ by monotonicity; hence $I(x) \leq z$.

			\item $I$ is Homogeneous of Degree 1: fix $x \in W^*$ and $\alpha >0$. 
			If $\alpha I(x) > I(\alpha x)$, then there is $w \in W$ such that
			$x\geq w $ and $\alpha  \tilde{I}(w)> I(\alpha x)$.
			Observe that $\alpha w\leq \alpha x$, $\alpha w \in W$ and so $\tilde{I}(\alpha w) = \alpha \tilde{I}(w)$, immediately leading to a contradiction; reversing the argument leads to a contradiction if $\alpha I(x) < I(\alpha x)$.
			
			\item $I$ is action invariant: $I(\alpha \phi  +(1-\alpha) g)= \alpha I(\phi )+ (1-\alpha) I(g)$ when $g = f_p$ for $p \in \Delta A$. 
			Notice that $w \in W$ $\iff$ $ \alpha w  +(1-\alpha) g \in W$, 
			and that if  $\phi \geq w$, then $\alpha \phi  +(1-\alpha) g \geq \alpha w +(1-\alpha) g$. The rest follows from $\tilde{I}$ being action invariant.	
			
			\item $I$ is supnorm continuous: Suppose not, so $x_n \rightarrow x$ in supnorm and, first, $\lim  \inf I(x_n)< I(x)$.
			There is $\epsilon>0$ and a sub-sequence, WLOG the whole sequence, such that $ I(x_n)+\epsilon<I(x)$ for all $n$.
			By definition, there exists $x \geq  w \in W$ such that $\tilde{I}(w)\geq I(x)-\epsilon/3$.
			Also, for $n$ large enough, $x_n\geq x -\epsilon/3$ in every state.
			Thus $x_n \geq w - \epsilon/3$, but then
			$I(x_n)\geq \tilde{I}(w -\epsilon/3)=\tilde{I}(w)-\epsilon/3 \geq I(x) -2 \epsilon/3$, a contradiction.
			Second, if  $\lim \sup I(x_n)> I(x)$, then there exists $\epsilon>0$  a sub-sequence, WLOG the whole sequence, such that $ I(x_n)>I(x)+\epsilon$ for all $x_n$. Pick $n$ such that $x\geq x_n-\epsilon/3$.
			There exist $x_n \geq  w \in W$ such that $\tilde{I}(w)\geq I(x_n)-\epsilon/3$.
			Then $x\geq w-\epsilon/3$ and $I(x)\geq \tilde{I}(w-\epsilon/3)\geq I(x_n)-2\epsilon/3>I(x)$, a contradiction.
		\end{itemize}
		
		
		To finalize the proof, I adapt the \cite{GilboaSchmeidler1989} (GS) arguments to construct a set of priors representing the preference as in GS but with the additional property that $\int  f_{(1,\langle a \rangle)} d\pi =\int f_{(1,\langle a \rangle)} d\pi'$ for all $\pi,\pi' \in \Pi$ for any  $a \in \mathcal{A}$.
		Let $W_A$ be the cone generated by $\{f_p:p \in \Delta(\mathcal{A})\}$.

		For any $\phi  \in W^*$ with $I(\phi)>0$, define $D_1=\{\psi \in W^*:I(\psi) > 1\}$ and \[D_2 = co(\{\psi \in W^*: \psi \leq a \text{, }I(a)=1 \text{, and }a\in W_A \} \bigcup \{\psi \in W^*: \psi \leq \phi/I(\phi)\}). \]
		To apply the GS arguments, I show that $D_1 \bigcap D_2 = \emptyset$.
		By $I$ action invariant and convexity of the constituent sets, any $d_2 \in D_2$ equals $\alpha a_1 +(1-\alpha) a_2$ where $a_1 \leq a$ for $a \in W_A$ and $I(a)=1$, $a_2 \leq \phi/I(\phi)$ and $\alpha \in [0,1]$
		Then, $I(d_2) \leq I(\alpha a +(1-\alpha) a_2)$ by WM, which equals $ \alpha I(a)+(1-\alpha)I(a_2)$ by action invariant, which is less than \[
		\alpha I(a) +(1-\alpha) I(\phi/I(\phi))=1
		\] by WM.
		Conclude $I(d_2) \leq 1$ for any $d_2 \in D_2$ and hence $D_1 \bigcap D_2 = \emptyset$.
		Moreover, note $1 \in D_2$ and $1$ in $cl(D_1)$.
		A separating hyperplanes argument gives a finitely additive measure $\pi_{\phi}$ such that $\int d_1 d\pi_{\phi} \geq 1 \geq \int d_2 d\pi_{\phi}$ for all $d_1 \in D_1$ and $d_2 \in D_2$.
		
		Applying the GS arguments shows that $\pi_{\phi}$ is a finitely additive probability measure, $I(\phi)=\int \phi d \pi_{\phi}$, and $ \int \psi d \pi_{\phi} \geq I(\psi) $ for all $\psi \in W^*$.
		This $\pi_{\phi}$ must have $\int a d \pi_{\phi} = I(a)$ for all $a \in W_A$, since $I(a/I(a))=1$ implies that $a/I(a) \in D_2$ and $1 \geq \int a/I(a) d\pi_{\phi}$.
		As in GS,  for $\Pi = \bar{co} \{ \pi_{\phi}:I(\phi)>0\}$, $p \succsim q$ if and only if \[
		\min_{\pi \in \Pi} \int f_pd\pi \geq \min_{\pi \in \Pi} \int f_q d\pi.
		\]
		completing the proof.
	\end{proof}
	
	\begin{proof}[Proof of Theorem \ref{thm:CC_rep}]
		By arguments as in EP, a unique coarsest correlation cover $\mathcal{U}$ exists.
		For any $C:\mathcal{A} \rightarrow \mathcal{U}$, let $f_p^C \in B_0(\Omega^\mathcal{U},\Sigma^{\mathcal{U}})$ be defined by\[
		f_p^C=\sum_{p( \langle a_{i}\rangle)>0)}   p\big(\langle a_{i}\rangle_{i=1}^{n}\big) u \left(\sum_{i=1}^{n}a_i \left( \omega^{C(a_{i})} \right) \right) 
		\]
		The result follows from the same arguments as in Theorem \ref{thm:basic_rep} if $f^{C_1}_p \geq f^{C_2}_q$ implies that $p \succeq q$ for any $C_1,C_2:\mathcal{A} \rightarrow \mathcal{U}$ with $a \in C_i(a)$ for $i=1,2$.
		
		Write $\Omega= \{1,\dots,K\}$  and $Na$ for $N$ copies of the action $a$, where $N$ is a positive integer. 
		For $x\in X$ and $B \in \mathcal{U}$ choose an action $\beta_{x}^{B,k}\in B$ so that $\beta_{x}^{B,k}(\omega)$ equals $x$ if $\omega = k$ and $0$ otherwise and define the corresponding event
		\[
		\mathcal{E}^{B,k,x}=\{\vec{\omega}\in\Omega^{\mathcal{A}}:\omega^{\beta
			_{x}^{B,k}}\in E_{B}^{k}\}.
		\]
		Note such actions exist because $B$ is rich. 
		Let $\Theta_{\varepsilon}=(-\epsilon,0)\bigcup (0,\varepsilon)$, i.e. an open interval of size $\varepsilon$ around $0$ that excludes $0$. 
		
		Restricted to $\Delta A$, $\succsim$ satisfies the Anscombe-Aummann axioms and it is easy to verify that there exists $\mu \in \Delta \Omega$ so that for $p,q \in \Delta A$, $p \succsim q$ if and only $U(p)\geq U(q)$ where $U(p)=\int_{\Omega} \mathbb{E}_p(\langle a \rangle) [u(a(\omega))] d \mu $.
		\begin{lem}
			\label{LemmaSmallBets} Suppose that there exist $x,y\in X$ such that $u(x+y)+u(0)\neq
			u(x)+u(y)$. There exists $\varepsilon>0$ such that for every collection $\{\beta_{x_1}^{B_1,k_1},\dots,\beta_{x_n}^{B_n,k_n}\}$ with $x_i \in \Theta_{\varepsilon}$, $B_i \in \mathcal{U}$, and $k_i \in \Omega$ for each $i$, and any $p \in \mathcal{F}$ there exists \[
			\pi_0 \in \arg \min_{\pi \in \Pi } \int f_p d\pi
			\] such that
			\begin{align}
			&\pi_0(\mathcal{E}^{B_i,k_i,x_i})= \mu (k_i) \label{eq:marg}\\
			k_i \neq k_j, \ B_i = B_j  \implies & \pi_0(\mathcal{E}^{B_i,k_i,x_i} \bigcap \mathcal{E}^{B_j,k_j,x_j}) = 0\label{pcrpf1}\\
			\text{and }  
			k_i = k_j, \ B_i = B_j  \implies & \pi_{0} \left( \mathcal{E}^{B_i,k_i,x_i} \bigcap \mathcal{E}^{B_j,k_j,x_j} \right)  = \mu(k_i)\label{pcrpf2}
			\end{align}
			for all distinct $i,j\in\{1,\dots,K\}$ and every $B \in \mathcal{U}$.
		\end{lem}
		
		In words,  for any $\theta$, there is minimizing probability measure $\pi_0$ with the following properties.
		Eq (\ref{eq:marg}) requires that the marginals of $\pi_0$ agree with $\mu $.
		Eq (\ref{pcrpf1}) implies that the DM believes it impossible that bets on distinct states in the same class pay off jointly.
		Eq (\ref{pcrpf2}) implies that if one bet on state $i$ pays off, then all bets on state $i$ in the same class pay off. 
		In sum, within the same understanding class, all the bets on one and only one of the elements of its finest partition pay off jointly.
		
		\begin{proof}
			[Proof of Lemma \ref{LemmaSmallBets}]
			Following the proof of  Lemma 4 from EP, for any non-zero $x',y'\in (-\varepsilon,\varepsilon)$ for $\varepsilon>0$ small enough, the absolute value of
			\begin{equation}
			u(Nx^{\prime}+ My^{\prime}+ z_{0})+ u(z_{0})
			-u(Nx^{\prime}+ z_{0}) -u(My^{\prime}+ z_{0})  \label{eq: non-linear}%
			\end{equation} 
			is sufficiently close to $u(x+y)+u(0)-u(x)-u(y)$ for some positive integers $N$ and $M$ and an appropriately chosen $z_0$. In particular, one can find $\varepsilon>0$ so that (\ref{eq: non-linear}) does not equal zero for every non-zero $x',y'\in (-\varepsilon,\varepsilon)$.	
			
			To ease notation, set $\beta_i=\beta^{B_i,k_i}_{x_i}$ and
			$\mathcal{E}^i=\mathcal{E}^{B_i,k_i,x_i}$.
			
			First, observe that for $\pi \in \Pi$, $\pi (\mathcal{E}^i) \geq \mu(k_i)$, since
			\[
			\mu(k_i)u(x_i)=I(\beta_i)=\min_{\pi \in \Pi}\pi (\mathcal{E}^i)u(x_i). \]
			Second, for any $p \in \Delta \mathcal{F}$,  there exists $\pi \in \arg\min_{\pi \in \Pi} \int f_p d\pi $ with $\pi (\mathcal{E}^i) = \mu (k_i)$ for all $i$.
			To see why, pick arbitrary  $p$ and note $\alpha I( f_p) +(1- \alpha)I(f_{\beta_i}) =I(\alpha f_p+(1- \alpha)f_{\beta_i})$ since $I$ is action independent for $i=1$.
			The former equals $\alpha I(f_p)+(1- \alpha) \mu(k_i)u(x_i)$.
			The latter equals $\int [\alpha f_p]d\pi+ \pi (\mathcal{E}^i)u(x_i)$ for some $\pi \in \Pi$.
			If $\pi (\mathcal{E}^i)>\mu (k_i)$, then $\int f_p d\pi <I(f_p)=\min_{\pi'\in \Pi} \int f_p d\pi'$, contradicting the definition of $I$.
			Conclude there is a minimizer with  $\pi (\mathcal{E}^1) =\mu(k_1)$.
			Now, suppose for $n$ there is  $\pi \in \arg\min_{\pi \in \Pi} \int f_p d\pi $ with $\pi (\mathcal{E}^i) = \mu(k_i)$ for $i<n$.
			Repeat the above arguments with $i=n$, but choose $\pi$ to be the minimizer claimed by the IH. Conclude that this minimizer must also have $\pi (\mathcal{E}^n) = \mu(k_n)$.
			Induction implies this must be the case for all $\pi(\mathcal{E}^i)$.
			Hence, for any $p$, there is a minimizer satisfying Equation (\ref{eq:marg})
			
			Third, claim that this minimizer can also be taken to have $\pi_0(\mathcal{E}^i \bigcap \mathcal{E}^j)=0$ when $B_i=B_j$ and $k_i \neq k_j$.
			There are a finite number of these pairs of events; order them arbitrarily.
			Assume (IH) that there is a minimizer $\pi_0$ for $p$ satisfying Eq (\ref{eq:marg}) and for which Eq (\ref{pcrpf1}) also holds for the first $n-1$ pairs.
			The base case holds by step 2.
			
			Let $(i,j)$ be  pair $n$ and $p$ be any given lottery.			 
			Since $x_i,x_j \in \Theta_{\varepsilon}$, by the above, there exists $N,M,z_0$ such that \[ u(N x_i+ M x_j+z_0)+u(z_0)-u(N x_i+z_0)-u(M x_j+z_0) = D \neq 0.\]
			Define lotteries  \[
			p_1 \equiv \left(\frac{1}{2}, \langle N\beta^{i},z_{0}\rangle; \frac{1}{2},\langle M\beta^{j},z_{0}\rangle \right)\]
			and \[p_2 \equiv  \left( \frac{1}{2},\langle N \beta^{i},M\beta^{j},z_{0}\rangle;\frac{1}{2}z_{0}\right).
			\] 
%
			Since $B$ understood,
			\[
			q_1=\frac12 p +\frac12 p_1
			\sim \frac12 p +\frac12 p_2=q_2
			\] 
			Since $I$ is action independent,
			\[I(f_{q_1})=\frac12 I(f_p)+\frac12 I(f_{p_1})\]
			and 	\[
			I(f_{p_1})=\mu(k_i)[u(Nx_i+z_{0})-u(z_{0})]+\mu(k_j)[u(Mx_j+z_{0})-u(z_{0})]+u(z_{0})
			\]
			By IH,   there exists $\pi_0 \in C$ satisfying (\ref{eq:marg}) so that 
			\begin{align*}
			I(f_{q_2})=&\int f_{q_2} d\pi_0\\
			=&\frac12  \int f_p d\pi_0+ \frac12 \pi_{0}(\mathcal{E}^{j} \bigcap \mathcal{E}^{i})[u(Nx_i+Mx_j+z_{0})-u(z_{0}
			)]+\frac12 u(z_{0})\\
			&\qquad+\frac12[\pi_{0}(\mathcal{E}^{i})-\pi_{0}(\mathcal{E}^{j} \bigcap \mathcal{E}^{i})][u(Nx_i+z_{0})-u(z_{0})] \\
			&\qquad+\frac12[\pi_{0}(\mathcal{E}^{j})-\pi_{0}(\mathcal{E}^{j} \bigcap \mathcal{E}^{i})][u(M x_j+z_{0})-u(z_{0})]\\
			=&\frac12  \int f_p d\pi_0 +\frac12 I(f_{p_1}) + \frac12 \pi_{0}(\mathcal{E}^{j} \bigcap \mathcal{E}^{i}) D
			\end{align*}
			If $\pi_{0}(\mathcal{E}^{j} \bigcap \mathcal{E}^{i})>0$, then $I(f_{q_1})\neq I(f_{q_2})$, contradicting the claimed indifference.
			Conclude the IH holds for the first $n$ pairs as well.
			Conclude by induction that there is a minimizer satisfying Eq (\ref{pcrpf1}) for any $p \in \Delta(\mathcal{F})$.
			
			Fourth, suppose $\beta_i,\beta_j \in B \in \mathcal{U}$.
			Let $b\in B$ be a bet yielding $x_i$ on $\Omega \setminus \{k_i \}$ and $0$ otherwise.
			\[
			\mathcal{E}^{b}=\{\vec{\omega}\in\Omega^{\mathcal{A}}:\omega^{b} \neq k_i \}.
			\]
			Because $B$ is understood, one has, for any $N \in \mathbb{N}$ and $z\in X$, that
			\begin{align*}
			\frac12p +\frac12 \left(\frac{1}{2},\langle N \beta_{i} , z \rangle; \frac{1}{2},\langle N b,z\rangle \right) 
			&\sim \frac12p +\frac12 \left( \frac{1}{2},\langle N \beta_i,N b, z \rangle; \frac{1}{2},\langle z \rangle \right)\\
			&\sim \frac12p +\frac12 \left( \frac{1}{2},\langle N x_i, z \rangle; \frac{1}{2},\langle z \rangle \right).
			\end{align*}
			By above, there is a minimizer satisfying Eqs. (\ref{eq:marg}) and (\ref{pcrpf1}), and
			similar arguments to those establishing Eq. (\ref{pcrpf1}) show the minimizer $\pi_0$ for  $f_p$ satisfying Eqs. (\ref{eq:marg}) and (\ref{pcrpf1}) can be taken to also satisfy
			\[
			\pi_{0}\left(  \mathcal{E}^{i}\bigcap\mathcal{E}^{b}\right)
			=\pi_{0}\left(  \mathcal{E}^{j}\bigcap\mathcal{E}^{b}\right)  =0.
			\]
			Picking  $N \in \mathbb{N}$ and $z\in X$ such that $u(z+Nx^{\prime})\neq u(z)$, one also has that
			\[
			\left[  \pi_{0}(\mathcal{E}^{i})+\pi_{0}(\mathcal{E}^{b})\right]
			(u\left(  N x_i+z\right)  -u\left(  z\right)  )= u\left(  N x_i+z\right)  -u\left(  z\right)
			\]
			and so
			\[
			\pi_{0}(\mathcal{E}^{i})+\pi_{0}(\mathcal{E}^{b})=1.
			\]
			The inclusion-exclusion formula gives that
			\[
			1\geq\pi_{0}\left(  \mathcal{E}^{i}\bigcup\mathcal{E}^{j}\bigcup\mathcal{E}^{b}\right)  =1+\pi_{0}(\mathcal{E}^{j}
			)-\pi_{0}\left(  \mathcal{E}^{i}\bigcap\mathcal{E}^{j}\right)
			\]
			and thus $\pi_{0}(\mathcal{E}^{j})=\pi_{0}\left(  \mathcal{E}%
			^{i}\bigcap\mathcal{E}^{j}\right)  $. 
			A symmetric	argument with $b^{\prime}$ defined using $x_j$ instead of $x_i$ shows $\pi_{0}(\mathcal{E}^{i})=\pi_{0}\left(  \mathcal{E}%
			^{i}\bigcap\mathcal{E}^{j}\right)  $.
			Inductively extending as above yields a minimizing $\pi_0$ satisfying Eq. (\ref{pcrpf2}).
		\end{proof}
		
		\begin{lem}
			\label{lem: betreplace} Given any $\varepsilon>0$ and profile $F=\langle
			a_{i}\rangle_{i=1}^{n}$ and allocation $C:\mathcal{A} \rightarrow \mathcal{U} $, there exist $\beta_{1},...,\beta_{T}\in\mathcal{A}$,
			$B_{1},...,B_{T}\in\mathcal{U}$ and $N_{1},...,N_{T}\in\mathbb{N}_{+}$ such
			that:
			\newline(i) for any $B_{j}$, $j=1,....,T$, there exists $a_i$ such that $C(a_i)=B_{j}$;
			\newline(ii) for any	$j=1,....,T$, $\beta_{j}=\beta_{x}^{B_{j},k}$ for some $k\in
			\{1,...,K\}$ and $x\in\Theta_{\varepsilon}$;
			\newline(iii) For any $C \in \mathcal{U}$ and all $\omega\in\Omega$,%
			\[
			\sum_{\{j:B_{j}=C\}}N_{j}\beta_{j}(\omega)=\sum_{\{i:C(a_i)=C\}}a_{i}(\omega).
			\]
		\end{lem}
		The proof of Lemma \ref{lem: betreplace} follows the same arguments as Lemma 5 of EP.
		
		Fix $\varepsilon$ as per Lemma \ref{LemmaSmallBets}. 
		Denote the profile $F^{\beta,C}$ for $\varepsilon$.
		For any $p=(p_i,F_i)_{i=1}^n$ and allocation $C:\mathcal{A} \rightarrow \mathcal{U} $, an induction proof using that each $C(a_i)$ is understood gives that \[p \sim (p_i,F^{\beta,C}_i)_{i=1}^n\equiv p^{\beta,C}.\]
		
		Now, to conclude, pick any $p,q\in \Delta \mathcal{F}$ and any allocations  $C_1,C_2:\mathcal{A} \rightarrow \mathcal{U} $ satisfying $f^{C_1}_{p} \geq f^{C_2}_{q}$.
		By Lemma \ref{LemmaSmallBets}, $I(f_{p^\beta})=\int f_{p^\beta} d\pi_0$ for some $\pi_0 \in \Pi$ satisfying Eqs. (\ref{eq:marg}), (\ref{pcrpf1}) and (\ref{pcrpf2}) for all the bets in the profiles in the supports of in $p^{\beta,C_1},q^{\beta,C_2}$.
		As in the proof of EP, only states that correspond to $f_{p^{\beta,C_1}}$ equaling  $f^{C_1}_{p}(\vec{\omega})$ and $f_{q^{\beta,C_2}}$ equaling $f^{C_2}_{q}(\vec{\omega})$ are possible according to $\pi_0$. Hence,  \[I(f_{p})=I(f_{p^{\beta,C_1}})=\int f_{p^{\beta,C_1}} d\pi_0  \geq  \int f_{q^{\beta,C_2}} d\pi_0  \geq  I(f_{q^{\beta,C_2}})= I(f_{q}),\]
		which implies that $p \succsim q$.
	\end{proof}
	
		\begin{proof}[Proof of Proposition \ref{prop: pref_simple}.]
		Necessity is trivial. Suppose $\succsim$ has a rich CCR $\left(u,\mu,\mathcal{U},\Pi\right)$ where $u$ is not a polynomial and $\mathcal{U}$ is finite.
		Since $u$ is not a polynomial, there exists $x,y,z$ so  that $u(x+z)+u(y+z)-u(x+y+z)-u(z) \neq 0$. To save notation, set $z=0$; adding $z$ to each of the profiles in the lotteries compared below covers the case where $z \neq 0$.
		To save notation, write $EF$ instead of $E\bigcap F$ for events $E,F \in \Sigma^{\mathcal{U}}$.
		
		First consider $K=u(x)+u(y)-u(x+y)-u(0) < 0$. 
		For $E \in \Sigma$, let $a_E^i=xE0 \in C_i$ and $b_E^j=0Ey \in C_j$ and $a^i_E+b^j_E =xEy\in C_1$.
		Then $V(a_E+b_E)=\mu(E) u(x)+\mu(E^c)u(y)$ and 
		there is $\pi \in \Pi$ so that 
		\begin{align*}
		V(\langle a_E^1,b_E^2 \rangle)=
		& [\pi(E_{C_1})-\pi (E_{C_1} E^c_{C_2})] u(x)+
		[\pi((E^c)_{C_2})-\pi (E^c_{C_1} E_{C_2})] u(y)+\\
		&[\pi (E_{C_1} E^c_{C_2})] u(x+y)+
		[\pi (E^c_{C_1} E_{C_2})] u(0)\\
		=   & \mu(E)u(x)+\mu(E^c)u(y)-\pi (E^c_{C_1} E_{C_2})K
		\end{align*}
		since $\pi(E_{C_1})=\pi(E_{C_2})=\mu(E)$, $\pi(E^c_{C_1})=\pi(E^c_{C_2})=\mu(E^c)$, and $$\pi (E_{C_1} E^c_{C_2})=\pi (E^c_{C_1} E_{C_2}) =\pi(E_{C_i})-\pi(E_{C_1}E_{C_2}).$$
		UM implies $V(a+b)\geq V(\langle a,b \rangle)$, which holds only if $\pi (E^c_{C_1} E_{C_2})= 0$ because $K<0$.
		Similarly, for $E_1,\dots,E_n \in \Sigma$ and $j_1,\dots,j_n,k_1,\dots,k_n \in \mathcal{U}$ so that $k_i \neq j_i$,
		there exists $\pi \in \Pi$ so that
		$$V \left( \left( \frac{1}{n}, \langle  a^{j_i}_{E_i}+b^{k_i}_{E_i} \rangle \right)_{i=1}^n \right)-V \left(\left( \frac{1}{n}, \langle a^{j_i}_{E_i},b^{k_i}_{E_i}\rangle \right)_{i=1}^n \right)=
		-\frac{1}{n}\sum_{i=1}^n  \pi \left( (E_i)_{C_{j_i}} (E_i^c)_{C_{k_i}} \right)K .$$
		UM holds only if $\pi \left( (E_i)_{C_{j_i}} (E_i^c)_{C_{k_i}} \right)=0$ for each $i$.
		Choosing events and indexes appropriately establishes the result.
		
		If instead $K=u(x)+u(y)-u(x+y)-u(0) > 0$ repeat instead with $c^i_E=yE0 \in C_i$ and $a^i_E+c^j_E=(x+y)E0 \in C_1$ replacing $b^i_E$, noting
		$V(a^i_E+c^j_E)=\mu(E) u(x+y)+\mu(E^c)u(0)$.
		and there is $\pi' \in \Pi$ so that 
		\begin{align*}
		V(\langle a^1_E,c^2_E \rangle)=& 
		[\pi'(E_{C_1})-\pi' (E_{C_1} E^c_{C_2})] u(x+y)+
		[\pi'((E^c)_{C_2})-\pi' (E^c_{C_1} E_{C_2})] u(0)+\\
		&[\pi' (E_{C_1} E^c_{C_2})] u(x)+
		[\pi' (E^c_{C_1} E_{C_2})] u(y)\\
		=&\mu(E)u(x+y)+\mu(E^c)u(0)+\pi' (E^c_{C_1} E_{C_2})K
		\end{align*}
		and UM again requires $\pi (E^c_{C_1} E_{C_2})=0$.
	\end{proof}
		\begin{proof}[Proof of Proposition \ref{prop: independent benchmark}]
		
		By risk aversion, $u(x)+u(y)-u(x+y)-u(0)\equiv K> 0$.
		Fix two understanding classes $C_1,C_2$, and consider bets $a^{j}_{E_i} \in C_j$ so that $a^j_{E_i}(\omega)=xE_i0(\omega)$ and $c^{j}_{F_i}\in C_j$ s.t. $c^{j}_{F_i}(\omega)=yF_i0(\omega)$ and lottery
		\[p_i=\left(\mu(E_i)\mu(F_i),x+y;\mu(E_i)(1-\mu(F_i)),x;(1-\mu(E_i))\mu(F_i),y;(1-\mu(E_i))(1-\mu(F_i)),0 \right).\]
		The bets $a^{1}_{E_i} ,c^{2}_{F_i}$ are potentially misperceived:
		$C^1,C^2$ are rich and understood but $C^1 \bigcup C^2$ is not, since $\mathcal{U}$ is the coarsest correlation cover. 
		Consider the lottery $$q_i=\frac13 \langle a^{C_1}_{E_i},c^{C_2}_{F_i} \rangle+ \frac23\left( \frac12 \langle a^{C_1}_{E^c_i} \rangle+\frac12 \langle c^{C_2}_{F^c_i} \rangle \right)$$
		and note that $f_{q_i} \in B_0(\Omega^\mathcal{U},\Sigma^{\mathcal{U}})$, as defined in the proof of Theorem \ref{thm:CC_rep}, has 
		$$
		f_{q_i}(\vec{\omega})=\frac13 \left( u(x)+u(y) +K\chi_{(E_i)_{C_1} (F_i)_{C_2}}(\vec{\omega}) \right)$$
		where $\chi_E$ is the indicator function of the set $E$.
		Pick $n=(\# \Omega)^2$ such lotteries $q_1,\dots,q_n$ for which the span of $\{f_{q_1},\dots,f_{q_n}\}$ equals all of the $\Omega^{C_1,C_2}$-measurable functions in $B_0(\Omega^\mathcal{U},\Sigma^{\mathcal{U}})$.
		
		Now suppose  $\mu^2 \notin \Pi$. 
		By the separating hyperplane theorem, there exists $g \in  B_0(\Omega^\mathcal{U},\Sigma^{\mathcal{U}})$ so that $\int g d\mu^2 < 0 \leq \int g d\pi$ for all $\pi \in \Pi$; in particular, it is strictly less than $\min_{\pi \in \Pi} \int g d\pi$.
		Rescale $g$ by a positive affine transformation so it belongs to $co \{f_{q_1},\dots,f_{q_n}\}$. Then there are $\alpha_1,\dots,\alpha_n \geq 0 $ so that $g=\sum \alpha_i f_{q_i}$  and $\sum \alpha_i =1$.
		
		By construction, 
		\[
		    q'_i \equiv \frac12 \langle a^{C_1}_{E^c_i} \rangle+\frac12 \langle c^{C_2}_{F^c_i} \rangle  \sim \left( \frac12\mu(E^c_i),x;\frac12 \mu(F^c_i),y;1-\frac12(\mu(E^c_i)+\mu(F^c_i)),0 \right)\equiv p'_i
		\]
		Let $p''_i = \frac13 p_i + \frac23 p'_i$, observing that
		\[
		V(\sum \alpha_i p''_i)=\int g d\mu^2\] and 
		$$V(\sum \alpha_i q_i)=\min_{\pi \in \Pi} \int g d\pi$$ 
		so $\sum \alpha_i q_i \succ \sum \alpha_i p_i$. 
        Applying Lemma \ref{lem:action_IA}, 
        \[
        \sum \alpha_i (\frac13 p_i +\frac23 p'_i)=\frac13\sum \alpha_i p_i +\frac23\sum \alpha_i p'_i\succ \frac13\sum \alpha_i \langle a^{C_1}_{E_i},c^{C_2}_{F_i} \rangle +\frac23\sum \alpha_i q'_i=\sum \alpha_i q_i  \]
		if and only if \[
		    \sum \alpha_i p_i \succ \sum \alpha_i \langle a^{C_1}_{E_i},c^{C_2}_{F_i} \rangle, \]
		which contradicts default to independence.
%
%
%
	\end{proof}

	\bibliographystyle{apa}
	\bibliography{correlation_concern}
\end{document}